\documentclass[12pt,a4paper]{article}

\usepackage{a4}
\usepackage{amsmath,amsthm,amsfonts}
\usepackage{amssymb,latexsym}
\usepackage[cp1250]{inputenc}
\usepackage[T1]{fontenc}
\usepackage{color}
\usepackage{url}
\usepackage{xypic}

\newcounter{mnotecount}[section]
\renewcommand{\themnotecount}{\thesection.\arabic{mnotecount}}
\newcommand{\mnote}[1]%{}
{\protect{\stepcounter{mnotecount}}$^{\mbox{\footnotesize  $%\!\!\!\!\!\!\,
      \bullet$\themnotecount}}$ \marginpar{\raggedright\tiny%\em
    $\!\!\!\!\!\!\,\bullet$\themnotecount: #1} }

\def\lin{\mathop{\rm span}}
\newcommand{\bw}{{\mathbf w}}
\newcommand{\bmd}{{\mathbf m \mathbf d}}
\newcommand{\bm}{{\mathbf m}}
\newcommand{\bd}{{\mathbf d}}

\newcommand{\wk}{\widehat{k}\mbox{}}
\newcommand{\conff}{\psi}
\newcommand{\der}{\wr\wr}

\newcommand{\dd}{\,\mathrm{d}}

\newcommand{\real}{\mathbb {R}}

\newcommand{\eps}{\varepsilon}
\newcommand{\thet}{\vartheta}
\newcommand{\vphi}{\varphi}
\newcommand{\gdwa}{\tilde{g}\mbox{}}
\newcommand{\const}{\mathrm{const}}

\newcommand{\zg}{{\mathring g}}
\newcommand{\zk}{{\mathring k}}
\newcommand{\zka}{{\mathring \kappa}}
\newcommand{\wx}{\widehat{x}\mbox{}}
\newcommand{\wg}{\widehat{g}\mbox{}}

\newcommand{\dtwo}{\Delta_\sigma}
\newtheorem{theorem}{Theorem}
\theoremstyle{definition}
\newtheorem{definition}{Definition}
\newtheorem{remark}{Remark}
\newtheorem{example}{Example}

\newcommand{\ee}{\end{equation}}

\newcommand{\eq}[1]{(\ref{#1})}
\newcommand{\be}{\begin{equation}}
\newcommand{\bel}[1]{\begin{equation}\label{#1}}

\newcommand{\arxiv}[1]{\url{http://arxiv.org/abs/#1}}

\begin{document}
%\setcounter{page}{0}

%%=================================

%\newpage

\title{Rigid spheres in Riemannian spaces}
\author{%
{ Hans-Peter Gittel$^{1}$, }
%\thanks{E-mail: gittel@mathematik.uni-leipzig.de}
%\\
Jacek Jezierski$^{2}$,
%\thanks{Email: \texttt{Jacek.Jezierski@fuw.edu.pl}.}
\\ Jerzy Kijowski$^{3}$, %\thanks{Email:\texttt{kijowski@cft.edu.pl}.}
% \\
Szymon Łęski$^{3,4}$.%\thanks{Email: \texttt{s.leski@nencki.gov.pl}.}
\\ %
\\{}$^{1}$Department of Mathematics, University of Leipzig, \\
Augustusplatz 10, 04109 Leipzig, Germany\\
\\ {}$^{2}$Department of Mathematical Methods in Physics\\
Faculty of Physics, University of Warsaw,\\
ul. Hoża 69, Warszawa, Poland \\
\\
{}$^{3}$Center for Theoretical Physics, Polish Academy of Sciences \\
al. Lotników 32/46, Warszawa, Poland
\\
\\
{}$^{4}$Nencki Institute of Experimental Biology,\\ Polish Academy of Sciences\\
ul. Pasteura 3, Warszawa, Poland}
\maketitle
\begin{abstract}

We define a special family of topological two-spheres, which we call ``rigid spheres'', and prove that there is a four-parameter family of rigid spheres in a generic Riemannian three-manifold whose metric is sufficiently close to the flat metric (e.~g.~in the external region of an asymptotically flat space). In case of the flat Euclidean three-space these four parameters are: 3 coordinates of the center and the radius of the sphere. The rigid spheres can be used as building blocks for various (``spherical'', ``bispherical'' etc.) foliations of the Cauchy space. This way a supertranslation ambiguity may be avoided. Generalization to the full 4D case is discussed. Our results generalize both the Huang foliations (cf. \cite{LHH}) and the foliations used by us (cf. \cite{JKL}) in the analysis of the two-body problem.

\end{abstract}

\section{Introduction}

In General Relativity Theory, the amount of gravitational energy (mass) contained in a portion $V \subset \Sigma$ of a Cauchy three-surface $\Sigma$ is assigned to its boundary $S=\partial V$, rather than to the volume $V$ itself (cf.~the notion of a ``quasi-local'' mass introduced by Penrose, \cite{QLM}). The above philosophy was also used in \cite{pieszy}, where important quasi-local observables (like, e.g., momentum, angular momentum or center of mass) assigned to a generic 2D surface (whose topology is that of $S^2$) have acquired a hamiltonian interpretation as generators of the corresponding canonical transformations of the (appropriately defined) phase space of gravitational initial data. Recently, we were able to define energy contained in an asymptotically Schwarzschild-de Sitter spacetime (cf.~\cite{CJK-2013}), and again the quasi-local, hamiltonian description of the field dynamics provided an adequate starting point for our analysis.

Typically, the $S^2$-spheres used for the quasi-local purposes come from specific spacetime foliations $\{ t=\ $const.; $r=\ $const.$\}$, where a specific choice of coordinates $t$ and $r$ plays the role of a gauge. In literature, gauge conditions based on 3D-elliptic problems have been mostly used (see e.g.~``traceless-transversal'' condition advocated by J.~York (see e.g. \cite{YorkTT})
or a ``p-harmonic gauge'' analyzed in \cite{positivity}). Important results have been obtained by Huisken and Ilmanen (cf.~\cite{inverse-mean-flow}) who used a parabolic gauge condition imposed for the radial coordinate $r$. The same gauge was also used by one of us (J.J., see \cite{JJimcf}) to prove stability of the Reissner--Nordstr\"{o}m
solution, together with a version of Penrose's inequality.

%which uses the parabolic gauge to show that small (but fully nonlinear) perturbation
%(in the asymptotic region) of initial data (Einstein + Maxwell) from
%Reissner--Nordstr\"{o}m solution leads to the growth of the total energy
%(hamiltonian) of the system  and in this way the stability of this solution is
%measured; moreover, some version of the Penrose's inequality is proved).

For purposes of the quasi-local analysis, these approaches exhibit an obvious drawback consisting in the fact that we do not control intrinsic properties of the surfaces $\{ r = \mbox{\rm const.} \}$ constructed this way. This feature was partially removed by Huang in \cite{LHH}, where new 3D foliations were thoroughly analyzed. Their fibers $\{ r = \mbox{\rm const.} \}$ are selected by a 2D-elliptic condition: $k=$\ const., where $k$ denotes the mean extrinsic curvature. In a generic Riemannian three-manifold $\Sigma$, the above equation admits a one-parameter family of ``spheres''. Physically, they are related to the ``center of mass'' of the geometry (cf. \cite{LHH}).

Unfortunately, the above condition is not stable with respect to small perturbations of the geometry. Indeed, in the (flat) Euclidean space $E^3$, this condition admits not ,,one-'' but a four-parameter family of solutions (parameterized e.g.~by the radius $R$ and the three coordinates of a center). Moreover, the exclusive use of the center of mass reference frame is often too restrictive for physical applications. In particular, it does not allow us to describe easily the momentum -- i.e.~the generator of space translations.

In the present paper we propose a new gauge condition, which is also 2D-elliptic but does not exhibit the above drawback. Indeed, in a generic Riemannian three-manifold our condition selects a four-parameter family of solutions, like in the Euclidean space $E^3$. Moreover, our condition is weaker than  ``$k=$\ const.'' (equivalent in the non-generic, Euclidean case, only). Topological two-spheres satisfying our condition will be called
``rigid spheres''. They can be organized in topologically different ways: not necessarily standard ``nested spheres foliations'', but also e.g.~``bispherical foliations'' which already proved to be very useful in the analysis of the two body problem\footnote{Initial data for the two black holes system can be easily obtained from the flat Euclidean geometry $E^3$ by two ``punctures''. Such a space admits the ``$k=$\ const.'' foliation only in the external region, far away from the two bodies. On the contrary, our ``rigid spheres'' can be organized into a ``bispherical system of coordinates'' which covers nicely the entire exterior of the two horizons.} (see \cite{JKL}). We expect that various such arrangements, with rigid spheres used as building blocks, will provide useful gauge conditions in General Relativity Theory.

The present paper is a part of a bigger project, where we construct ``spheres'' which are rigid not only with respect to 3D, but also with respect to 4D deformations. More precisely, an eight-parameter family of similar ``rigid spheres'' will be constructed
in a generic four-dimensional Lorentzian spacetime.  In the present
paper we limit ourselves to the 3D Riemannian case. It turns out, however, that our construction can be generalized to the
entire pseudo-Riemannian spacetime $M$, instead of the Riemannian
Cauchy three-space $\Sigma \subset M$. The idea of this extension
is to mimic the case of the flat Minkowski space, where all
possible round spheres, embedded in all possible flat subspaces
$\Sigma$ of $M$, form an eight-parameter family. All of them can be obtained from a single one by the action of the product of the one-parameter group of dilations (changing
the size of $S$) and the ten-parameter Poincar\'e group,
quotiented by the three-parameter rotation group.
The 4D version of our construction will take into account not only the external curvature of $S$, but
also its torsion (in Section 2.5 we give a short outline of this construction, which will be presented in detail in a subsequent paper). The rigid spheres obtained this way will form an eight-parameter family and will be used
to construct useful coordinate systems not only on a given Cauchy
surface $\Sigma$, but also in the entire spacetime. The main advantage of such a construction consists in its rigidity at infinity. We very much hope to be able to eliminate supertranslations and to reduce the symmetry group of the ``Scri'', otherwise infinite dimensional, to the finite dimensional one.

The construction which we propose in the present paper is based on the following idea. Given a surface $S$ satisfying the rigid sphere condition, consider its infinitesimal deformations. They may be parameterized by sections of the normal bundle $T^\perp S$. If we want our condition to admit a four-parameter family of solutions, like in the flat case, its linearization must admit a four-parameter family of deformations. This means that we are not allowed to constrain the complete information about the mean curvature $k$: four real parameters describing $k$ must be left free. In the flat case these four parameters which have to be left free are: the mean value (or the {\em monopole part}) of $k$, which is responsible for the size of $S$, and its {\em dipole part} (which vanishes exceptionally in flat case due to Gauss-Codazzi equations). The {\em dipole part} of the deformation is related to the group of translations. In fact, possible motions of a metric sphere are described by the group of Euclidean motions, quotiented by the subgroup of rotations which form the group of internal symmetries of every particular sphere $S$.

To implement the above idea in a non flat case, an intrinsic, geometric notion of a multipole expansion on an arbitrary Riemannian, topologically $S^2$-surface is proposed in Section 2. This construction is our main technical tool and we very much believe in its universal validity, going far beyond the purposes of the present paper. Section 2 is completed with the definition of a rigid sphere.

Section 3 contains formulation and the proof of theorem 3: a generic Riemannian three-space admits a four-parameter family of rigid spheres. Our proof is relatively simple, but is valid in the ``weak field region'' only. This is sufficient for purposes of the quasi-local analysis of gravitational energy (in fact, the idea originates from our analysis of interaction between two black holes, cf.~\cite{JKL}). Further development concerning strong fields will be given elsewhere.

Finally, discussion concerning less known (but necessary) technical results, like specific spectral properties of the Laplace operator on $S^2$ or the second variation of area, has been shifted to the Appendix.

\section{Equilibrated spherical coordinates. Mul\-ti\-pole calculus on distorted spheres}
\subsection{Conformally spherical coordinates}

Let $S$ be a differential two-manifold, diffeomorphic to the two-sphere
$S^2\subset {\mathbb R}^3$ and equipped with a (sufficiently
smooth) metric $g$. Coordinates $(\thet, \varphi) = (x^A)$,
$A=1,2$, defined on a dense subset of $S \setminus \ell$, where
$\ell$ is topologically a line interval, will be called {\em
conformally spherical coordinates} if they have the same range of
values as the standard spherical coordinates on $S^2 \subset
{\mathbb R}^3$ and, moreover, if the corresponding metric tensor
$g_{AB}$ is conformally equivalent to the standard round metric on
$S^2$, i.e. the following formula holds:
\begin{equation}\label{conf0}
    g_{AB} = \conff \cdot \sigma_{AB} \ ,
\end{equation}
where $\conff$ is a (sufficiently smooth) function on $S$ and
\begin{equation}\label{eta}
    \sigma_{AB} =
    \left( \begin{array}{cc}
     1 &  0  \\   0 &  \sin^2 \thet
    \end{array} \right) \ .
\end{equation}

\begin{remark} Conformally spherical coordinates always exist (cf. \cite{isocoord}).
It is easy to see that there is
always a six-parameter freedom in the choice of such coordinates.
More precisely, if $(\thet, \varphi )$ are conformally spherical
coordinates then $(\widetilde\thet, \widetilde\varphi )$ are also
conformally spherical if and only if they may be obtained from $(\thet, \varphi )$ {\em via} a conformal transformation of
$S^2\subset {\mathbb R}^3$.
\end{remark}

\begin{example}{A ``proper'' conformal transformation,
i.e.~which is not a rotation:} Let ${\bf n} \in S$ and $\tau > 0$ be
a positive number. Using appropriate rotation, choose conformally
spherical coordinates $(\thet ,\varphi)$ in such a way that ${\bf
n}$ is a north pole, i.e.~the coordinate $\thet$ vanishes at
${\bf n}$. Define
\begin{equation}\label{F}
    F_{{\bf n},\tau}(\thet ,\varphi) =
    (\widetilde\thet ,\widetilde\varphi) \ ,
\end{equation}
where
\begin{equation}\label{alpha}
    \widetilde\thet := 2 \arctan \left(\tau \cdot \tan \frac
    \thet 2 \right), \ \ \ \ \ \ \ \ \
    \widetilde\varphi := \varphi \ ,
\end{equation}
or, equivalently,
\begin{equation}\label{theta-half}
    \tan \frac {\widetilde\thet} 2 =
    \tau \cdot \tan \frac \thet 2 \ .
\end{equation}

For the fixed point ${\bf n}$ these transformations form a
one-parameter group\footnote{In stereographic coordinates calculated with respect to the south pole this group is the dilatation group: $\zeta \rightarrow \tau \zeta$.}:
\begin{equation}\label{group}
    F_{{\bf n},\tau}\circ
    F_{{\bf n},\sigma} =
    F_{{\bf n},\tau\sigma} \ ,
\end{equation}
generated by the vector field:
\begin{equation}\label{field}
    \left. \frac {\mbox{\rm d}}{\mbox{\rm d} t}\right|_{t=1}
    F_{{\bf n},t}(\thet ,\varphi) =
    \left. \frac {\mbox{\rm d}}{\mbox{\rm d} t}\right|_{t=1}
    \left[2 \arctan \left( t \cdot \tan \frac \thet 2 \right)
    \right] \frac {\partial}{\partial\thet}=
    \sin \thet \frac {\partial}{\partial\thet} \ ,
\end{equation}
which is the (minus) gradient of the function $z = \cos \thet$.

In particular, $F_{{\bf n},1}= {\mathbb I}$ (the identity map)
for every ${\bf n}$. Moreover, equation \eqref{theta-half} implies the following identity:
\begin{equation}\label{identity}
    F_{{\bf - n},\tau} = F_{{\bf n},\frac 1\tau} \ .
\end{equation}

Using (\ref{alpha}) and (\ref{theta-half}) we may easily derive
the following formula:
\begin{equation}\label{dtheta}
    \mbox{\rm d} \thet = \frac {\mbox{\rm d} \thet}
    {\mbox{\rm d} {\widetilde\thet}} \ \
    {\mbox{\rm d} {\widetilde\thet}}=
    \tau \frac {1 + \tan^2 \frac {\widetilde\thet} 2}{\tau^2
    + \tan^2 \frac {\widetilde\thet} 2}
    \ \mbox{\rm d} {\widetilde\thet} \ .
\end{equation}
Similarly, we may prove:
\begin{eqnarray}
  \sin \thet  &=& \frac {\sin\thet}{\sin{\widetilde\thet}}
  \ \ \sin {\widetilde\thet}
  %= \frac {2\tan \frac \thet 2}{1 + \tan^2 \frac \thet 2}
 % \cdot \frac{1 + \tan^2 \frac
 % {\widetilde\thet} 2}{2\tan \frac {\widetilde\thet} 2}
 % \ \ \sin {\widetilde\thet}
  \nonumber \\
   &=& \frac {2\frac 1\tau \tan \frac {\widetilde\thet} 2}
   {1 + \frac 1{\tau^2}\tan^2 \frac {\widetilde\thet} 2}
  \cdot \frac{1 + \tan^2
  \frac {\widetilde\thet} 2}{2\tan \frac {\widetilde\thet} 2}
  \ \ \sin {\widetilde\thet}
  =\tau \frac {1 + \tan^2 \frac {\widetilde\thet} 2}
  {\tau^2 + \tan^2 \frac {\widetilde\thet} 2}
  \ \ \sin{\widetilde\thet} \ . \label{sin}
\end{eqnarray}
As a conclusion we obtain:
\begin{equation}\label{conf}
    (\mbox{\rm d} \thet )^2 + \sin^2 \thet (\mbox{\rm d}\varphi)^2
    = h^2 \left[
    (\mbox{\rm d} \widetilde\thet )^2 +
    \sin^2 \widetilde\thet (\mbox{\rm d}\varphi)^2
    \right] \ ,
\end{equation}
where
\begin{equation}\label{f-conf}
    h = \tau \frac {1 + \tan^2 \frac {\widetilde\thet} 2}
  {\tau^2 + \tan^2 \frac {\widetilde\thet} 2} \ ,
\end{equation}
which proves the conformal character of the transformation.
Indeed, we have
\begin{equation}\label{conf4}
    g_{AB}\mbox{\rm d}x^A \mbox{\rm d} x^B =
    \conff \left[
    (\mbox{\rm d} \thet )^2 + \sin^2 \thet (\mbox{\rm d}\varphi)^2
    \right] =
    \conff h^2 \left[
    (\mbox{\rm d} {\widetilde\thet} )^2
    + \sin^2 {\widetilde\thet} (\mbox{\rm d}\varphi)^2
    \right] \ .
\end{equation}
Hence, $({\widetilde\thet},\varphi)$ are conformally spherical
coordinates if $(\thet , \varphi )$ were.
\end{example}

\subsection{Barycenter of a conformally spherical system}

Given a system of conformally spherical coordinates on $S$,
consider the corresponding three functions:
\begin{eqnarray}
  x &:=& \sin \thet \cos\varphi \ , \label{x}\\
  y &:=& \sin \thet \sin\varphi \ , \label{y}\\
  z &:=& \cos \thet \label{z} \ .
\end{eqnarray}
We have, therefore, a mapping ${\mathbf D}\!: \, ]0,\pi[\times]0,2\pi[ \, \mapsto {\mathbb R}^3$, given by:
\begin{equation}\label{dipolD}
{\mathbf D}(\thet,\varphi)=
\left(\begin{array}{c}
D^1(\thet,\varphi)\\ D^2(\thet,\varphi)\\ D^3(\thet,\varphi)
\end{array}\right) =
\left(\begin{array}{c}
x \\
y \\
z \end{array}\right) \ .
\end{equation}
The following vector
\begin{equation}\label{X}
    {\bf X} =
    \left( \begin{array}{c}
     < x >  \\   < y > \\ < z >
    \end{array} \right) \in {\mathbb R}^3 \ ,
\end{equation}
where by $< f >$ we denote the average (mean value) of the
function $f$ on $S$, i.e. the number
\begin{equation}\label{average}
    <f> := \frac{\int_S f \sqrt{\det g}\ \dd^2x}
    {\int_S \sqrt{\det g}\ \dd^2 x}
 \ ,
\end{equation}
will be called
a ``barycenter'' of the system $(\thet,\varphi )$ on $S$.
Of course, we have $\| {\bf X} \| \leq 1$, because of the H\"older inequality:
\[
   \| {\bf X} \|^2 \, =  \, <x>^2 + <y>^2 + <z>^2 \ \le \ <x^2> + <y^2> + <z^2> \,=\, 1 \ .
\]

\begin{example} Consider the proper conformal transformation
(\ref{alpha}) and calculate the new barycenter
\begin{equation}\label{tildeX}
    \widetilde{\bf X} =
    \left( \begin{array}{c}
     < \widetilde x >  \\   < \widetilde y > \\ <\widetilde z >
    \end{array} \right) \in {\mathbb R}^3 \ ,
\end{equation}
where
\begin{eqnarray}
  \widetilde x &:=& \sin \widetilde\thet \cos\varphi \ , \nonumber\\
  \widetilde y &:=& \sin \widetilde\thet \sin\varphi \ , \nonumber\\
  \widetilde z &:=& \cos \widetilde\thet \nonumber \ .
\end{eqnarray}
The trigonometric identity:
\begin{equation}\label{cos}
    \cos \thet = \frac {
    1 - \tan^2 \frac {\thet} 2}
    {1 + \tan^2 \frac {\thet} 2} \ ,
\end{equation}
implies:
\begin{equation}\label{tg^2}
    \tan^2 \frac {\thet} 2 = \frac {1 - \cos\thet}
    {1 + \cos\thet} =\frac {1 - z}{1 + z} \ .
\end{equation}
Hence, formula (\ref{theta-half}) implies:
\begin{equation}\label{tg-3}
    \frac {1 - \widetilde z} {1 + \widetilde z}=
    \tau^2 \frac {1 - z} {1 + z} \ ,
\end{equation}
or, equivalently,
\begin{equation}\label{z-tilde}
    \widetilde z = \frac
    {1 + z - \tau^2 (1-z)}
    {1 + z + \tau^2 (1-z)}
    \ .
\end{equation}
Moreover, formula (\ref{sin}) and its inverse:
\begin{equation}\label{sin-inv}
    \sin {\widetilde\thet}
  =\tau \frac {1 + \tan^2 \frac {\thet} 2}
  {1 + \tau^2 \tan^2 \frac {\thet} 2}
  \,  \sin{\thet}
  = \tau \frac {1 + \frac {1 - z}{1 + z}}
  {1 + \tau^2 \frac {1 - z}{1 + z}}
  \,  \sin{\thet}
  = \frac {2 \tau \sin\thet}{1 + z + \tau^2 (1-z)}\ ,
\end{equation}
give %\mnote{\color{red} Napisać wyjaśnienie recenzetowi!}
\begin{eqnarray}
  \widetilde x &:=& \frac {2 \tau}{1 + z + \tau^2 (1-z)}
  \,  x \ , \label{x-tilde}\\
  \widetilde y &:=& \frac {2 \tau}{1 + z + \tau^2 (1-z)}
  \, y  \label{y-tilde} \ .
\end{eqnarray}
To calculate mean values of the functions (\ref{x-tilde}),
(\ref{y-tilde}) and (\ref{z-tilde}) we do not need to pass to new
coordinates $(\widetilde\thet ,\varphi)$, but we may use, as
well, old coordinates $(\thet ,\varphi)$. But we see that for
$\tau \rightarrow 0$ we have $\widetilde x \rightarrow 0$,
$\widetilde y \rightarrow 0$, $\widetilde z \rightarrow 1$. The
Lebesgue theorem implies, therefore, that for $\tau \rightarrow 0$
we have
\begin{equation}\label{X-limit}
    \widetilde{\bf X} =
    \left( \begin{array}{c}
     < \widetilde x >  \\   < \widetilde y > \\ <\widetilde z >
    \end{array} \right) \longrightarrow
    \left( \begin{array}{c}
     < 0 >  \\   < 0> \\ < 1 >
    \end{array} \right) = {\bf n} \ .
\end{equation}
\end{example}

\subsection{Equilibrated spherical coordinates}
\begin{definition}
Conformally spherical coordinate system $(\thet,\varphi)$ is
called {\em equilibrated}, if its barycenter vanishes: ${\bf X} =
0\in {\mathbb R}^3$.
\end{definition}

\begin{remark}
If there are two equilibrated spherical systems on
$S$ then they are related by a rotation.
\end{remark}

\begin{theorem}
 Each metric tensor on $S$ admits a unique (up to
rotations) equilibrated spherical system.
\end{theorem}

\begin{proof} Given a metric tensor $g$ on $S$, choose first any
system of conformally spherical coordinates $(\thet,\varphi)$ on
$S$ and consider the corresponding identification of its points
with the points of $S^2 = \partial K(0,1) \subset {\mathbb R}^3$.
Consider now the mapping
\begin{equation}\label{calF}
    {\mathbb R}^3 \supset K(0,1)\ni {\bf N} \rightarrow
    {\cal F}({\bf N})\in K(0,1) \subset {\mathbb R}^3 \ ,
\end{equation}
given for ${\bf N} \ne 0$ by the following formula
\begin{equation}\label{calF-1}
    {\cal F}({\bf N}):= \widetilde{\bf X}_{{\bf n},\tau} \ ,
\end{equation}
where the latter is the barycenter of the coordinates
$(\widetilde\thet,\widetilde\varphi)$ obtained from
$(\thet,\varphi)$ by the proper conformal transformation
(\ref{F}) with
\begin{equation}\label{n-tau}
    {\bf n} := \frac {\bf N}{\| {\bf N} \| }
\end{equation}
and
\begin{equation}\label{n-tau-1}
    \tau := 1 - \| {\bf N} \| \ .
\end{equation}
For ${\bf N}=0$ formula \eqref{n-tau} has no sense, but then \eqref{n-tau-1} gives $\tau = 1$ and, whence, equation \eqref{theta-half} implies that the corresponding transformation (\ref{F}) reduces to identity, no matter which vector ${\bf n}$ do we choose. Consequently, we define ${\cal F}({\bf 0})$ as the barycenter of the original coordinates $(\thet,\varphi)$.
Obviously, ${\cal F}$ defined this way is continuous. Moreover, for $\| {\bf N} \|
=1$ we have ${\cal F}({\bf N}) = {\bf N}$ due to \eqref{n-tau-1} and
\eqref{X-limit}. This means that ${\cal F}$ reduces to the identical mapping when
restricted to the boundary $S^2 = \partial K(0,1) \subset K(0,1)$. Consequently, there must be a point ${\bf N}_0$ which
solves equation ${\cal F}({\bf N}_0) = 0$. This completes the existence proof. To prove the uniqueness, let us suppose that there is
another solution: ${\cal F}({\bf N}_1) = 0$. Consider now the
conformal transformation $F_{{\bf n}_1,\tau_1} \circ F_{{\bf
n}_0,\tau_0}^{-1}$. Since the proper conformal transformations do
not form any subgroup of the group of all conformal
transformations, we cannot assume that it is again a proper
transformation. But it may be decomposed into a product of
rotations and a proper conformal transformation:
\begin{equation}\label{m}
    F_{{\bf n}_1,\tau_1} \circ F_{{\bf n}_0,\tau_0}^{-1} =
    {\cal O}_1 \circ F_{{\bf m},\tau} \circ {\cal O}_0 \ ,
\end{equation}
where ${\cal O}_1$ and ${\cal O}_0$ are rotations. Denote by
$(\thet_0,\varphi_0)$ the spherical coordinates obtained from
$(\thet,\varphi)$ by the transformation $F_{{\bf n}_0,\tau_0}$
and then rotation ${\cal O}^{-1}_0$. Similarly, denote by
$(\thet_1,\varphi_1)$ the ones obtained from $(\thet,\varphi)$
by $F_{{\bf n}_1,\tau_1}$ and then by rotation ${\cal O}^{-1}_1$.
Because a rotation does affect equilibration of coordinates, both
systems $(\thet_0,\varphi_0)$ and $(\thet_1,\varphi_1)$ are
equilibrated. But the latter may be obtained from the former by a
proper conformal transformation $F_{{\bf m},\tau}$. We shall prove
that this is impossible unless $\tau = 1$ or, equivalently,
transformation $F_{{\bf m},\tau}$ is trivial (identical). For this
purpose consider, for each value of $\tau$, the linear function
$z_\tau$. Without any loss of generality we may assume that
\begin{equation}\label{m-1}
    {\bf m} = \left(
    \begin{array}{c}
     0  \\   0 \\  1
    \end{array} \right)
\end{equation}
(if this is not the case, it is sufficient to perform an
appropriate rotation of coordinates). Formula (\ref{z-tilde})
implies the following relation:
\begin{equation}\label{z-par}
    z_\tau = \frac
    {1 + z_0 - \tau^2 (1-z_0)}
    {1 + z_0 + \tau^2 (1-z_0)}
    \ .
\end{equation}
Hence
\begin{equation}\label{derivative}
    \frac {\mbox{\rm d}}{\mbox{\rm d} \tau} z_\tau =
    \frac {- 4 \tau \left(1-z_0^2\right)}
    {\left[ 1 + z_0 + \tau^2 (1-z_0)\right]^2} \le 0
    \ ,
\end{equation}
and it vanishes only at a single point $z_0 = 1$. Consequently,
its mean value:
\begin{equation}\label{derivative-mean}
    < \frac {\mbox{\rm d}}{\mbox{\rm d} \tau} z_\tau >
\end{equation}
is strictly negative. This implies that starting from $\tau = 1$
(which corresponds to the identity mapping $F_{{\bf m},1}$) and
moving towards the actual value $\tau < 1$, the ``$z$''-component
of the vector $\widetilde{\bf X}_{{\bf m},\tau}$ is strictly
increasing. It vanishes at the beginning because
$(\thet_0,\varphi_0)$ is equilibrated. Hence, it must be strictly
positive at the end. This means that the final system
$(\thet_1,\varphi_1)$ cannot be equilibrated unless $\tau = 1$
and, therefore, both systems coincide.
\end{proof}

Different equilibrated spherical systems of coordinates form, therefore, a three-dimensional family. They can be parameterized by the position of a fixed point ${\bf n} \in S$ (north pole) and the geographic longitude of a fixed point ${\bf m} \in S$ (Greenwich). More precisely: given two points ${\bf n}, {\bf m} \in S$, ${\bf n} \ne {\bf m}$, there is a unique equilibrated spherical system $(\thet,\varphi)$ of coordinates on $S$, such that $\vartheta$ vanishes at ${\bf n}$ and $\varphi$ vanishes at ${\bf m}$.

Combining these observations with classical results (cf.~\cite{isocoord}), we obtain the following

\begin{theorem}\label{unique}
Let $S$ be a differential two-manifold, diffeomorphic to the two-sphere
$S^2\subset {\mathbb R}^3$ and equipped with a metric $g$ of class $C^{(k,\alpha)}$. For every pair ${\bf n}, {\bf m} \in S$, ${\bf n} \ne {\bf m}$, there is a unique equilibrated spherical system $(\thet,\varphi)$ of coordinates on $S$, such that $\vartheta$ vanishes at ${\bf n}$ and $\varphi$ vanishes at ${\bf m}$, and the metric components $g_{AB}$ are of the same class $C^{(k,\alpha)}$.
\end{theorem}
\noindent Here, $C^{(k,\alpha)}$ is a
 H{\"o}lder space $C^{k, \alpha} (S^2)$, defined for  $1 \leqslant k \in {\mathbb N}$ and $0 < \alpha  < 1$. The space consists of those functions on $S^2$ which have continuous derivatives up to order $k$ and such that the $k$-th partial derivatives are H{\"o}lder continuous with exponent $\alpha$. This is a locally convex topological vector space.

 The H{\"o}lder coefficient of a function $f$ is defined as follows:
   \[ | f |_{C^{0,\alpha}} = \sup_{x,y \in S^2 ,\, x \neq y} \frac{| f(x) - f(y) |}{|x-y|^\alpha} \ . \]
The function $f$ is said to be (uniformly) H{\"o}lder continuous with exponent $\alpha$ if $| f |_{C^{0,\alpha}}$ is finite. In this case the H{\"o}lder coefficient can be used as a seminorm.

The H{\"o}lder space $C^{k,\alpha}(S^2)$ is composed  of functions whose derivatives up to order $k$ are bounded and the derivatives of the order $k$ are H{\"o}lder continuous. It is a Banach space equipped with the norm
 \[   \| f \|_{C^{k, \alpha}} = \|f\|_{C^k}+\max_{| \beta | = k} | D^\beta f |_{C^{0,\alpha}} \, , \]
where $\beta$ ranges over multi-indices and
 \[   \|f\|_{C^k} = \max_{| \beta | \leq k} \, \sup_{x\in S^2} |D^\beta f (x)| \, . \]

\subsection{Rigid spheres in a Riemannian three-space}

Given a manifold $S$ equipped with a metric tensor $g$, there is a
three-dimensional space of ``linear functions'' uniquely defined
on $S$ as linear combinations of
functions (\ref{x}--\ref{z}), calculated in any equilibrated spherical
system of coordinates $(\thet,\varphi)$. We denote this space by
${\cal M}^3$. By ${\cal M}^4$ we denote the space spanned by
${\cal M}^3$ and the constant functions on $S$.
Linear functions (\ref{x}--\ref{z}) on ${S}$
are eigenfunctions of the Laplace operator\footnote{By $\Delta_{\sigma}$ we denote the usual Laplace operator for the unit-sphere metric (\ref{eta}).}  $\Delta_{\sigma}$, with
the eigenvalue equal to $-2$, i.e.
%\[
  $\Delta_{\sigma} X^i = -2 X^i$,
%\]
where we denote $x=X^1$, $y=X^2$, $z=X^3$.
Let us denote by $\mbox{\rm d}\sigma:=\sin\thet\, {\rm d}\thet \, {\rm d}\varphi$
the measure associated with the metric $\sigma_{AB}$.
\begin{definition}
\label{Def0} Let $f \in L^2(S,\mbox{\rm d}\sigma)$. The projection of
$f$ onto the subspace of constant functions:
\begin{equation}\label{mon}
    P_m(f) := \frac{1}{4\pi}\int_{{ S}}
  f \, \mbox{\rm d}\sigma
\end{equation}
will be called the {\em monopole part} of $f$, whereas the
projection onto \\  ${\cal M}^3 = \lin \{ X^1,X^2,X^3 \}$:
\begin{equation} \label{eq1c}
  P_d(f) := \sum_{i=1}^3 \left(X^i {\int_{{ S}} X^i f
  \, \mbox{\rm d}\sigma \over \int_{{ S}} (X^i)^2 \, \mbox{\rm d}\sigma }\right)
\end{equation}
will be called the {\em dipole part} of $f$. In addition, we set
\begin{equation}
%%\begin{gather}
\label{eq2c} {\cal M}^4 := \lin \{ 1 \}
 \oplus {\cal M}^3
 = \lin \{1,X^1, X^2 ,X^3  \} \, ,
\end{equation}
and  $P_{md}(f):= P_{m}(f) + P_{d}(f) \in {\cal M}^4$
denotes the {\em mono-dipole part} of  $f$.
\end{definition}

The above structure enables us to define the multipole
decomposition of the functions defined on a topological sphere $S$
 in terms of eigenspaces of the Laplace operator associated with
the metric $\sigma_{AB}$. If $h$ is a function on $S$, then by
$h^\bm:= P_m(h)$ we denote its monopole (constant) part, by
$h^\bd:=P_d(h)$ --- the dipole part (projection to the eigenspace of the Laplacian
with eigenvalue $-2$). By $h^{\bw} := (I-P_{md})(h)= h -h^\bm- h^\bd$ we
denote the ``wave'', or mono-dipole-free, part of $h$, $h^{\bd\bw}
:= (I-P_{m})(h)= h -h^\bm= h^{\bd}+h^{\bw}$,
 and finally $h^{\bmd}:=P_{md}(h)= h^\bm + h^\bd$.

\begin{remark}
Mutually orthogonal projectors $P_{md}$ and $P_w :=(I-P_{md})$ are, of course, continuous, when considered as operators in the Hilbert space $L^2(S,\mbox{\rm d}\sigma)$. For our purposes we have to consider them as operators in the Banach space $C^{(k,\alpha)}$. Here, no ``orthogonality'' is defined. Nevertheless, both operators are again continuous projectors. They define an isomorphism:
\[
   C^{(k,\alpha)} \cong C_{md}^{(k,\alpha)} \times C_w^{(k,\alpha)} \ ,
\]
where $C_{md}^{(k,\alpha)} = P_{md} (C^{(k,\alpha)})\equiv {\cal M}^4$ and $C_{w}^{(k,\alpha)} = P_{w} (C^{(k,\alpha)})$. Hence, a function $f \in C^{(k,\alpha)}$ is uniquely characterized by its mono-dipole part $f^{\bm\bd}$ and the remaining ``wave'' part $f^\bw$, i.e. we have: $f = (f^{\bm\bd},f^{\bw})$.
\end{remark}

 \begin{definition}
Let $\Sigma$ be a Riemannian three-manifold
and let $S \subset \Sigma$ be a submanifold homeomorphic with $S^2 \subset {\mathbb R}^3$. We say that $S$ is a \emph{rigid sphere} if its mean extrinsic curvature $k$ satisfies $k \in {\cal M}^4$, i.e. if the following equation holds:
\begin{equation}\label{rig}
    k^{\bw} = 0 \ .
\end{equation}

\end{definition}

\subsection{The 4-D spacetime case -- an outline}

Definition of a rigid sphere in a Lorenzian four-manifold is more complicated: to control ``rigidity'' of a sphere, we must take into account more geometry.
For this purpose we consider the extrinsic curvature vector of $S$: $k^a = k^a_{AB} {g}^{AB}$, where $k^a_{AB}$ denotes the external curvature tensor of $S$ (here, $a,b$ are indices corresponding to the
subspace orthogonal to $S$ whereas $A,B$ label coordinates on $S$). Moreover, we consider its torsion:
\begin{equation}\label{ell}
    \ell_A = ({\bf m} | \nabla_A {\bf n}) \ ,
\end{equation}
where
\begin{equation}\label{n}
    {\bf n} := \frac {\bf k}{\| {\bf k} \| } \ ,
\end{equation}
${\| {\bf k} \| } = \sqrt{k^a g_{ab} k^b}$,
and ${\bf m}$ is a vector orthogonal to both ${\bf k}$ and $S$.

\begin{definition}
Let $M$ be a Lorenzian four-manifold (a generic
curved spacetime) and let $S \subset M$ be a {\em spacelike} submanifold homeomorphic with $S^2 \subset {\mathbb R}^3$. We say that $S$ is a
\emph{rigid sphere} if ${\bf k}=(k^a)$ is spacelike and the following
two conditions are satisfied:
\begin{eqnarray}
  {\| {\bf k} \| } &\in& {\cal M}^4 \ , \\
  \nabla_A \ell ^A &\in& {\cal M}^3 \ .
\end{eqnarray}

\end{definition}

In this paper we limit ourselves to the purely Riemannian 3D-setting.
The general, pseudo-riemannian case will be analyzed in a
subsequent paper.\\

\begin{example}%\noindent\underline{Example}:
{\it Rigid spheres in a four-dimensional Minkowski
spacetime and in Euclidean three-space.}
%The rigid spheres in Minkowski spacetime:
\\
Let ${ M}_0$ be the flat Minkowski spacetime, i.e.~the space
${\mathbb R}^4$ parameterized by the Lorentzian coordinates
$(x^\alpha)=(x^0, \ldots ,x^3)$ and equipped with the metric $\eta
= (\eta_{\alpha \beta}) = {\rm diag} (-1,1,1,1)$ (Greek indices
run always from 0 to 3). \\
Consider in ${ M}_0$ a {\em round sphere}, i.e.~the
two-dimensional submanifold defined by
\[
  S_{T,R} := \Big\lbrace x \in {\mathbb R}^4 \
  \Bigl| \ x^0 = T \ , \ \sum_{i=1}^3 (x^i)^2 = R^2 \Big\rbrace
   \ ,
\]
where the time $T\in {\mathbb R}$ and the sphere's radius $R>0$ are fixed.
It may be easily verified that the
submanifold fulfills the following conditions:
\begin{eqnarray}
  \sqrt{k^a g_{ab} k^b} = \tfrac {2}{R} &\in& {\cal M}^4 \ , \\
  \nabla_A \ell ^A =0 &\in& {\cal M}^3 \ ,
\end{eqnarray}
hence each round sphere $S_{T,R}$ in Minkowski spacetime ${ M}_0$
is a rigid sphere.
Using Poincar\'e symmetry group of ${ M}_0$, it is easy to check that there is an
8-parameter family of such spheres. Indeed, fixing the value of $R$,
a 7-parameter family remains left. All of them may be obtained
from a single sphere, say $S_{0,R}$, by the action of the
10-parameter Poincar\'e group. Because the three-parameter subgroup of rotations corresponds to internal symmetries of $S_{0,R}$, we are left with $7$ parameters only. The parameter $R$ corresponds to the dilation group.
Hence, we have 8 ($=10-3+1$) parameters.

In Euclidean three-space (represented by a slice $\{ x^0 = 0\}$ in ${M}_0$)
the family of rigid spheres reduces to four-parameter family of such spheres,
where $4=3+1$ -- three translations plus dilation
(or similarity transformations minus rotations $4=7-3$).
Each round sphere in Euclidean three-space is a rigid sphere because
its mean extrinsic curvature $k=-\tfrac {2}{R} \in {\cal M}^4$.
\end{example}

\section{Existence of rigid spheres in a Riemannian space}

Let $\Sigma$ be a three-dimensional Riemannian manifold. Let
$S\subset \Sigma$ be a two-manifold diffeomorphic to the unit
sphere $S^2\subset \real^3$.
We consider the following problems:  1) Can we deform $S$ in such a
way that the resulting submanifold becomes a rigid sphere? 2) How
many of such deformations exist in a vicinity of $S$?

To parameterize these deformations we introduce in a neighbourhood
of $S$ a Gaussian system of coordinates $(u,x^A)$. Here, by
$(x^A)$, $A=1,2$, we denote any coordinate system on $S$, whereas
$u$ is the arc-length parameter along the ``$\{ x^A=\const.\}$'' geodesics
starting orthogonally from $S$. The three-metric takes, therefore,
the form
    \be\label{Gcoord}
    g = \dd u^2 + g_{AB}(u,
    x^A)\dd x^A \dd x^B\;.
    \ee

Suppose, moreover, that coordinates $(x^A)=(\thet, \vphi)$ are conformal and equilibrated on $S$. This means that we have
    \be
    \zg_{AB}\dd x^A \dd x^B =
    \conff\cdot(\sigma_{AB}\dd x^A \dd x^B ) \; ,
    %= \conff\cdot(\dd\thet^2 + \sin^2\thet \dd \vphi^2)\; ,
    \ee
where
    \be
    \zg_{AB} := g_{AB}(0,x^A)
    \ee
is the induced two-metric on $S$, $\sigma$ is the ``round''
two-metric on the Euclidean unit sphere:
\be \sigma_{AB}\dd x^A \dd x^B =
    \dd \thet^2 + \sin^2\thet \dd\vphi^2\ ,
    \ee
and the function $\conff$ is dipole-free ($\conff^{\bd} = 0$).
Second fundamental form of $S$ is given by:
    \be
    \zk_{AB} = -\frac12 g_{AB,u}\ .
    \ee
Its trace does not need to belong to the space ${\cal M}^4$ of
mono-dipole-like functions, i.e.~the surface $S$ does not need to
be a rigid sphere. We are looking for such deformations of $S$,
for which the resulting surface fulfills already the rigidity
condition.

Any deformation of $S$ which is sufficiently small may be uniquely
parameterized by a function $\tau=\tau(x^A)$, such that the deformed
surface $S_\tau$ is given by:
    \bel{s-xi}
    S_\tau = \{ (u,x^A)\,|\,
    u=\tau(x^A)\}\, .
    \ee
The surface $S_\tau$  carries the induced metric:
\be\label{gonSxi}
    g \vert_{S_\tau} = \left[\dd \tau(x^A)\right]^2 + g_{AB}\left(\tau(x^C),
    x^C\right)\dd x^A \dd x^B
    =: g_{AB}\left(x^C\right)\dd x^A \dd x^B\, ,
    \ee
where
\be\label{explicite-g}
   g_{AB}\left(x^C\right) = (\partial_A \tau )( \partial_B \tau ) +
   g_{AB}\left(\tau(x^C),
    x^C\right) \, .
\ee
Here, we use the same coordinate system $(x^A)$, which was previously used for $S$. However, these coordinates do not need to be neither conformally spherical nor equilibrated. To verify that the deformation $\tau$
was successful, i.e.~that $S_\tau$ is a rigid sphere, we have to pass to an equilibrated system of spherical coordinates, say
$\wx^A$,  on $S_\tau$. To make this choice unique, we use the north pole: $\bf{n}:=\{ \thet = 0\}$, and the ``Gulf of Guinea'':  ${\bf m}:= \{ \thet = \frac \pi 2 ; \varphi = 0 \}$ to get rid of the rotation non-uniqueness (cf. Theorem \ref{unique}). This way we obtain an equilibrated version $\wg_{AB}$ of the metric \eqref{explicite-g}. Finally, we calculate the extrinsic curvature $k$ and check whether or not its wave part $k^{\bw}(S_\tau)$ satisfies condition $k^{\bw}(S_\tau)= 0$.

The idea of our paper may, therefore, be sketched as follows. We begin with a metric \eqref{Gcoord} which is of the class $C^{(k,\alpha)}$. The above construction defines a continuous mapping:
\begin{equation}\label{mapping-F}
    C_{md}^{(k+1,\alpha)} \times C_w^{(k+1,\alpha)} \ni (\tau^{\bm\bd},\tau^{\bw} ) = \tau \longrightarrow F(\tau):= k^{\bw} \in C_w^{(k-1,\alpha)} \, .
\end{equation}
Indeed, the resulting metric in a neighbourhood of $S_\tau$ is obtained from $g$ and the first derivatives of $\tau$. The function $\tau$ being of the class $C^{(k+1,\alpha)}$, the metric obtained this way is again of the class $C^{(k,\alpha)}$. Due to Theorem \ref{unique}, its equilibrated version $\wg_{AB}$ is again of the same class. Finally, the extrinsic curvature $k$ is obtained, using first derivatives of this metric. Hence, the result is of the class $C^{(k-1,\alpha)}$ and the entire procedure is continuous.

Now, rigid spheres are those, which satisfy equation:
\begin{equation}\label{implicit}
    F(\tau) = 0 \ .
\end{equation}
We are going to prove that, for a generic metric $g$, which is sufficiently close to the flat metric,  the above equation defines an implicit function:
\begin{equation}\label{mapping-H}
   {\cal M}^4 \equiv C_{md}^{(k+1,\alpha)} \ni \tau^{\bm\bd} \longrightarrow H(\tau^{\bm\bd})
    \in  C_w^{(k+1,\alpha)}  \ ,
\end{equation}
such that
\begin{equation}\label{tozsamosc}
    F ( \tau^{\bm\bd} , H(\tau^{\bm\bd}) ) \equiv 0 \ ,
\end{equation}
or, equivalently, that $S_{( \tau^{\bm\bd} , H(\tau^{\bm\bd}) )}$ is a rigid sphere. The main result of our paper follows
as: %a corollary:
\begin{theorem}\label{main}
 Generically (i.~e.~if the metric is sufficiently close to the flat metric, e.~g.~in the external region of an asymptotically flat space) there exists a four-parameter family of rigid spheres in a neighbourhood of a given two-sphere $S \subset \Sigma$, corresponding to the four-parameter family of mono-dipole functions $\tau^{\bm\bd}$ on $S$.
\end{theorem}
\subsection{Infinitesimal deformations of spheres}
\label{sc:defsph}

To prove existence of the implicit function \eqref{tozsamosc} it is sufficient to show that, given a mono-dipole deformation $\tau^{\bm\bd}$,  the partial derivative of $F$ with respect to the ``wave-like'' deformation $\tau^{\bw}$ is an isomorphism of $C_w^{(k+1,\alpha)}$ onto $C_w^{(k-1,\alpha)}$.

For this purpose, we analyze the infinitesimal, linear version of the construction discussed above. Consider, therefore, a transversal deformation $\tau=\tau(x^A)$ of $S \subset \Sigma$ and a small deformation parameter $\eps$:
    \bel{sxi}
    S\rightsquigarrow S_\tau = \{ (u,x^A)\,|\,
    u=\eps\tau(x^A)\}\;.
    \ee
Under such transformation the induced metric
changes in the following way:
    \be\label{g-g0}
    g_{AB} - \zg_{AB} = - 2\eps\tau\zk_{AB}+O(\eps^2)\;.
    \ee
Even if the initial system of coordinates was equilibrated, the
transformed metric does not need to be conformally spherical. The
non-sphericality of the metric must be, therefore, compensated by
a change of coordinates. Its infinitesimal version is described by a tangential (with respect to $S$) deformation
    \be\label{xtox}
    \wx^A = x^A - \eps\xi^A\;.
    \ee
Under such coordinate transformation the metric
changes as follows:
 \be
  \wg_{AB} = g_{AB} - \pounds_{2\eps {\vec \xi}}
  \ g_{AB} \;,
 \ee
where the last term represents the Lie derivative of the metric
$g_{AB}$ with respect to the vector field ``$- \eps\xi^A$'' on $S$. But,
according to \eqref{g-g0}, the difference between $g_{AB}$ and
$\zg_{AB}$ is already of the first order in $\eps$. Hence, if we
replace it by the Lie derivative of the metric $\zg_{AB}$, the
error will be of the second order in $\eps$. Using the Killing
formula for the Lie derivative of the metric, we finally obtain:
 \be
  \wg_{AB} = g_{AB} + 2\eps \xi_{(A||B)} +O(\eps^2) \;,
 \ee
and the covariant derivative ${}_{||A}$ is taken with respect to
the original metric $\zg_{AB}$. Hence, we have:
 \be
 \wg_{AB} -
 \zg_{AB} = -2\eps\tau\zk_{AB} +
 2\eps\xi_{(A||B)} +O(\eps^2)\;.
 \ee
Let us decompose the above equation into the trace and the
trace-free parts, calculated with respect to $\zg_{AB}$ (we omit
the terms of order $\eps^2$ and higher):
    \be  \wg_{AB}-\zg_{AB} =
    \left( \eps \xi^C{}_{||C}-\eps\tau\zk \right)\zg_{AB} -
    2\eps\tau\zka_{AB} + 2 \eps \left( \xi_{(A||B)} - \frac12
    \xi^C{}_{||C}\zg_{AB}\right)\;,
    \ee
where
 \be
 \zka_{AB} := \zk_{AB} - \frac12 \zk
 \zg_{AB}
 \ee
is the traceless part of $\zk_{AB}$. We want $\wg_{AB}$ to be
conformally spherical, i.e. $\wg_{AB} = \alpha \cdot \zg_{AB}$.
This implies:
    \be
    \left(1-\eps\tau\zk - \alpha + \eps \xi^C{}_{||C}\right)\zg_{AB}
    - 2\eps\tau\zka_{AB} + 2 \eps \xi_{(A||B)} - \eps
    \xi^C{}_{||C}\zg_{AB} = 0\;.
    \ee
The trace part of this equation defines uniquely the value of
$\alpha$:
    \bel{alpha1}
    \alpha = 1-\eps\tau\zk  + \eps \xi^C{}_{||C} \ ,
    \ee
whereas the trace-free part reduces to:
    \bel{rxiA0}
    \xi_{A||B} + \xi_{B||A} -
    \xi^C{}_{||C}\zg_{AB} = 2 \tau \zka_{AB}\;.
    \ee

It is convenient to rewrite equation \eq{rxiA0} in terms of the
``round'' unit-sphere geometry~$\sigma_{AB}$. For this purpose we
use the following conventions: components of a vector (i.e.~an
object having {\em upper indices}) are the same in both geometries
$\sigma_{AB}$ and $\zg_{AB}= \psi \sigma_{AB}$. Components of a
co-vector ({\em lowered indices}) are denoted as follows:
    \be
    \xi^\sigma_A = \sigma_{AB}
    \xi^B \ \  , \ \ \
    \xi_A = \zg_{AB} \xi^B = \psi \sigma_{AB} \xi^B = \psi \xi^\sigma_A \, .
    \ee
The covariant derivative with respect
to $\sigma_{AB}$ will be denoted by ${}_{\der A}$,
e.g.
    %\be
    $\xi^\sigma_{A\der B}$. %\;.
    %\ee
Equation \eq{rxiA0} can be easily rewritten as:
    \bel{i12}
     \xi^\sigma_{A\der B}
    + \xi^\sigma_{B\der A} - \xi^C{}_{\der C}\sigma_{AB}
     = 2 \frac \tau\conff \zka_{AB}\;.
    \ee
The left-hand side of this equation defines a
mapping from the space of vector fields on the unit sphere to the
space of trace-free rank 2 tensor fields. The kernel of this
mapping consists of the dipole fields\footnote{A vector field on the sphere may
be uniquely decomposed into the sum of a gradient and a co-gradient. These two components are represented by the corresponding two scalar functions: the
divergence and the curl. The multipole expansion of a vector field
is uniquely defined by the multipole expansion of these two
functions.}. The ``Fredholm alternative'' argument shows that the
operator on the left-hand side defines an isomorphism between the
space of dipole-free vector fields on the unit sphere and the
space of trace-free rank 2 tensor fields (see also \cite{peel:JJ}).
This isomorphism (in metric $\sigma$) will be denoted by $i_{12}$.
Hence, the wave part of $\xi^A$ is implied uniquely by equation
\eq{i12} (see Appendix):
    \be \label{wave-xi}
    \xi^{\bw A} = i_{12}^{-1}\left(2 \frac\tau\conff
    \zka_{AB}\right)\ ,
    \ee
whereas the dipole part of $\xi^A$, i.e.~the field $\xi^{\bd A}$,
remains arbitrary.

The above choice of the wave-like component of the tangential deformation $\xi^{\bw A}$ guarantees that the new
coordinate system $\wx^A$ is conformally spherical. We would like
it to be also: 1) equlibrated and 2) satisfying conditions related to the two fixed points ${\bf n}$ and ${\bf m}$. These conditions mean that the field $\xi$ has to vanish at the north pole ${\bf n}$ and that its $\varphi$-component vanishes at ${\bf m}$. The above $3 + 3 = 6$ conditions fix uniquely the total dipole-part of the tangential (to $S$) deformation $\xi^A$. This way the continuous mapping which assigns uniquely the tangential deformation $\xi^A$ to its transversal component $\tau$ has been defined.

\subsection{The infinitesimal change of the extrinsic curvature}

Now, we are going to calculate the infinitesimal
change of the wave part $k^\bw$ of the mean curvature\footnote{First
variations of the {\em total} mean curvature $k$ is known in the literature as the second variations of area, cf. e.g. \cite{svararea}.~See also discussion in the Appendix.}, i.e.~derivative of the mapping \eqref{mapping-F} with respect to the ``wave-like'' deformation $\tau^{\bw}$.  We have $k=\gdwa^{AB} k_{AB}$, where
$\gdwa^{AB}$ denotes the inverse of the two-metric $g_{AB}$ (whereas $g^{AB}$ denotes the corresponding components of the inverse three-metric.) The simplest way to calculate this change is to use a coordinate system $(\omega, x^A)$,
adapted to the deformed surface:
    \be
    \omega = u - \eps\tau(x^A)\;,
    \qquad \mathrm{i.e.} \;\; S_\tau = \{\omega = 0\}\ ,
    \ee
and the formula:
    \be
    k_{AB} = \frac{1}{\sqrt{
    g^{\omega\omega}}} \Gamma^\omega{}_{AB}\;.
    \ee
The three-metric $g$ takes now the following form:
    \be
    g=
    \dd\omega^2 + 2\eps\tau_{,A} \dd\omega \dd x^A +
    g_{AB}\dd x^A \dd x^B + O(\eps^2)\;.
    \ee
This implies $g^{\omega\omega} = 1+O(\eps^2)$ and, consequently,
\be
    k_{AB} =  \Gamma_\omega{}_{AB}  + g^{\omega C} \Gamma_C{}_{AB} + O(\eps^2) =
    \frac 12 \left( g_{\omega A || B} + g_{\omega B || A} - g_{AB , \omega}\right) +
    O(\eps^2) \; ,
    \ee
where we treat the ``shift vector'' $g_{\omega A}=\eps\tau_{,A}$ as a covector field on $S_\tau$. The first two terms combine to $\eps\tau_{||AB}$, whereas the last one: $g_{AB , \omega} (S_\tau)$ can be approximated by the quantity $g_{AB , \omega} (S)= -2\zk_{AB}$ plus the derivative of this object. Finally,
 we have
    \be
    k_{AB}
    = \zk_{AB} +\eps\tau \zk_{AB,u} + \eps\tau_{||AB} +
    O(\eps^2)\, .
    \ee
Since the derivative $g_{AB , \omega}$ of the metric $g_{AB}$ is described by $-2\zk_{AB}$, the derivative of its inverse $\gdwa^{AB}$ is described by $+2\zk^{AB}$. Hence, we have:
\be\label{g-g0up}
    \gdwa^{AB} - \zg^{AB} =  2\eps\tau\zk^{AB}+O(\eps^2)\, ,
    \ee
and, consequently:
\be\label{varku} k =\gdwa^{AB} k_{AB}= \zk + \eps\tau \partial_u \zk
    +\eps\tau^{||A}{_A} + O(\eps^2)\, .
    \ee

The quantity  $\tau \partial_u \zk +\tau^{||A}{_A}$ describes already the second variation of area (see Appendix), i.e. the derivative $\nabla_\tau k$. However, to calculate the derivative of the mapping \eqref{mapping-F}, we have to select its wave part $k^\bw$. For this purpose we have to pass to the conformally spherical, equilibrated coordinates $\wx^A$, given by formula \eqref{xtox}. Infinitesimal change of the scalar function $k$ with respect to this deformation is given by formula:
\[ \wk=k-\eps\xi^A k_{,A}+O(\eps^2)\, .\] Hence, we get:
    \be
    \frac1\eps\left(\wk - \zk\right) =
    \tau \partial_u \zk +\tau^{||A}{_A}
     - \xi^A \zk_{,A}  +  O(\eps)\, ,
    \ee
or, equivalently (cf.~Appendix),
\be
    \frac1\eps\left(\wk - \zk\right) =
    \tau (R^u{}_u + \zk^{AB} \zk_{AB})+\tau^{||A}{_A}
     - \xi^A \zk_{,A}  +  O(\eps)\, ,
    \ee
where $R^u{}_u=R({\rm d} u,{\partial \over \partial u}) $ is the component of the Ricci tensor.

\subsection{Proof of the Theorem \ref{main}}

The last formula gives, finally, the value of the derivative of the mapping \eqref{mapping-F}. When restricted to the subspace of wave (i.e.~mono-dipole-free) deformations, it gives us:
 \bel{xitok}
   C_w^{(k+1,\alpha)} \ni \tau \mapsto \left[\tau \left(R^u{}_u + \zk^{AB}
    \zk_{AB}\right) + \tau^{||A}{_A}
     -\,  \xi^A \zk_{,A}\right]^{\bw} \in C_w^{(k-1,\alpha)} \ .
    \ee
The above linear operator is, obviously, continuous. In particular, the vector field $\xi^A$ is given by formula \eqref{wave-xi}, together with the accompanying vanishing conditions at ${\bf n}$ and ${\bf m}$.

If the space $\Sigma$ is flat (Euclidean) and $S$ is a standard (rigid) sphere of radius $r$, then we have:
\begin{equation}\label{baza}
    \zg_{AB}=r^2\sigma_{AB} \, ; \quad \zk_{AB}=-r\sigma_{AB} \, ;
    \quad \zk_{,A}=0 \, ; \quad R^u{_u}=0 \, .
\end{equation}
Hence, the above operator reduces to:
\begin{eqnarray}  \nonumber
   \tau^{\bw} & \mapsto &  \left[\tau \left(R^u{}_u + \zk^{AB}
    \zk_{AB}\right) + \tau^{||A}{_A} -  \xi^A \zk_{,A}\right]^{\bw}=
     \frac1{r^2}\left[(\Delta_\sigma+2)\tau\right]^{\bw}  \\ & &  =\frac1{r^2}(\Delta_\sigma+2)(\tau^{\bw}) \ ,
    \label{xitoklw} \end{eqnarray}
which is obviously an invertible mapping from $C_w^{(k+1,\alpha)}$ to $C_w^{(k-1,\alpha)}$. But the mapping \eqref{xitoklw} depends in a continuous way upon the geometry (metric and curvature) of $S$. This implies that it remains invertible for sufficiently small deformations of the geometry.
This is the case e.g. of a sufficiently big ``coordinate sphere'' defined as follows:
\[
  S_{\vec{x}_0,R} := \Big\lbrace x \in {\mathbb R}^3 \
  \Bigl|\ \sum_{i=1}^3 (x^i - x_0^i)^2 = R^2 \Big\rbrace
   \ ,
\]
in an asymptotically flat $\Sigma$.

We say, that $\Sigma$ is asymptotically flat if there is a coordinate chart $(x^k)$ covering the exterior of a compact domain $D \subset \Sigma$ and such that
\[
  g_{kl} = \delta_{kl} + h_{kl} \ ,
\]
where $h$ vanishes sufficiently fast at infinity.
%This means that the above statement remains true for ``coordinate spheres'' defined as follows:
%\[
%  S_{\vec{x}_0,R} := \Big\lbrace x \in {\mathbb R}^3 \
%  \Bigl|\ \sum_{i=1}^3 (x^i - x_0^i)^2 = R^2 \Big\rbrace   \ ,
%\]
%if the radius $R$ is sufficiently big.
In that case $\Sigma \setminus D$ admits a four-parameter family of rigid spheres, similarly as in the case of the flat metric.

\section{Conclusions}
The main technical ingredient of this paper is the intrinsic, coordinate invariant definition of the ``multipole expansion'' of a function defined on a Riemannian two-manifold, diffeomorphic with $S^2$. This enables us to select a finite-dimensional family of ``rigid spheres''. The dipole part $k^\bd$ of the curvature parameterizes the position of the center of such a sphere with respect to the center of mass. In particular, $k^\bd= 0$ corresponds to the spheres, which are centered at the center of mass. Properties of such a foliation have been analyzed in \cite{LHH}. General topologically spherical coordinates, having property that surfaces $\{ r = $ const.$\}$ are rigid, do not admit {\em supertranslations} ambiguity at space infinity. This way symmetries of the ``tangent space at infinity'' reduce to a finite-dimensional one. The 4D version of our results, valid for a generic four-dimensional Lorenzian spacetime, which will be presented in the subsequent paper, will do the same job for the symmetry group of the Scri.

\section*{Acknowledgements}
This research was supported by Polish Ministry of Science and Higher
Education (grant Nr N N201 372736) and by Narodowe Centrum Nauki (grant DEC-2011/03/B/ST1/02625).
%This research was supported by \ldots Nr N N201 372736.
SŁ was supported by Foundation for Polish Science.

\appendix
\section{Appendix}

\subsection{The dipole part of traceless symmetric part}
%of the covariant derivative of the co-vector vanishes}
 The kernel of the mapping
\[ \xi^\sigma_A \mapsto \xi^\sigma_{A\der B}
    + \xi^\sigma_{B\der A} - \sigma^{CD}\xi^\sigma_{C\der D}\sigma_{AB} \]
defined by the left-hand side of the formula (\ref{i12})
consists of the dipole fields. This is a simple consequence of the following
observations.
\begin{itemize}
\item In case of the unit sphere
the Hodge decomposition
$\xi=\dd\alpha+\delta\beta+h$ of the covector $\xi$ on a compact manifold
%simplifies
does not contain the harmonic part,
%and it has vanishing harmonic part
i.e. harmonic one-form $h$ vanishes ($\dd h=0=\delta h$ implies $h=0$).
The topology of the unit sphere
(triviality of the corresponding cohomology class)
 cancels the harmonic part and
we can always represent $\xi$ as follows
\begin{equation}\label{xiab}
 \xi_A = \alpha_{,A} + \varepsilon_A{^B}\beta_{,B} \, ,
\end{equation}
where functions $\alpha$ and $\beta$ are defined up to a constant
but their gradients are unique.
\item The purely dipole covector $\xi$ simply means that the potentials
$\alpha$ and $\beta$ are purely dipole functions:
$\alpha=a_iX^i$, $\beta=b_iX^i$, where $a_i$, $b_i$ are real constants.
\item Direct computation for dipole functions $X^i$ enables one to check
the following identity:
$X^i_{\der AB}=-X^i\sigma_{AB}$, hence for any dipole function $\alpha$
\begin{equation} \label{aab}
\alpha_{\der AB}=-\alpha\sigma_{AB} \, .
\end{equation}
\item Formulae (\ref{xiab}) and (\ref{aab}) %applied to $\xi_{A\der B}$
give
\[ \xi_{A\der B}= -\alpha\sigma_{AB}-\beta\varepsilon_{AB} \, ,\]
hence the traceless symmetric part of $\xi_{A\der B}$ vanishes.
\end{itemize}
\subsection{The isomorphism between covector fields and symmetric traceless tensors on $(S^2,\, \sigma_{AB})$}
Let us consider the following diagram:
\[
\begin{array}{ccccccccc}
V^0_{k+2}\oplus V^0_{k+2} & \stackrel{i_{01}}{\longrightarrow} & V^1_{k+1} &
\stackrel{i_{12}}{\longrightarrow} & V^2_{k} &
\stackrel{i_{21}}{\longrightarrow} & V^1_{k-1}
& \stackrel{i_{10}}{\longrightarrow} & V^0_{k-2} \oplus V^0_{k-2} \\
\Big\downarrow\vcenter{\rlap{$\scriptstyle Fl$}} &  &
\Big\downarrow\vcenter{\rlap{$\hat{} $}} &  &
\Big\downarrow\vcenter{\rlap{$\hat{} $}} &  &
\Big\downarrow\vcenter{\rlap{$\hat{} $}} &  &
\Big\downarrow\vcenter{\rlap{$\scriptstyle Fl$}}  \\
V^0_{k+2}\oplus V^0_{k+2} & \stackrel{i_{01}}{\longrightarrow} & V^1_{k+1} &
\stackrel{i_{12}}{\longrightarrow} & V^2_{k} &
\stackrel{i_{21}}{\longrightarrow} & V^1_{k-1}
& \stackrel{i_{10}}{\longrightarrow} & V^0_{k-2} \oplus V^0_{k-2}
\end{array}
\]
where the mappings and the spaces are defined as follows:
\[ i_{01}(f,g)=f_{\der A}+\varepsilon_A{^B}g_{\der B} \, ,\]
\[ i_{12}(v)= v_{A\der B}+ v_{B\der A}-\sigma_{AB}v^C{_{\der C}} \, , \]
\[ i_{21}(\chi)= \chi_A{^B}{_{\der B}} \, ,\]
\[ i_{10}(v)=\left( v^A_{\der A}, \varepsilon^{AB}v_{A\der B} \right) \, ,\]
\[ Fl(f,g)=(g,f) \, , \quad {\hat v}_A=\varepsilon_A{^B}v_{B} \, , \quad
 {\hat\chi}_{AB}=\varepsilon_A{^C}\chi_{CB} \, ,\]
 \noindent
$V^0_k$ -- scalars on $S^2$ belonging to H\"older space $C^{k,\alpha}$ \, ,\\
 $V^1_k$ -- covectors on $S^2$ belonging to H\"older space $C^{k,\alpha}$\, ,\\
$V^2_k$ -- symmetric traceless tensors on $S^2$ belonging to H\"older space $C^{k,\alpha}$.\\
 Denote by $\dtwo$ the Laplace operator on $S^2$ and by $SH^l$ the space of
spherical harmonics of degree $l$, ($f\in SH^l \Longleftrightarrow
 {\dtwo}f= -l(l+1)f$).
The following equality
\[ i_{10} \circ i_{21} \circ i_{12} \circ i_{01} = {\dtwo}({\dtwo}+2) \]
shows that if we restrict ourselves to the spaces $\overline V^0:=V^0
\ominus [SH^0\oplus SH^1 ] = (I-P_{md})V^0$
 (${\dtwo}({\dtwo}+2)\overline V^0 =\overline V^0$) and
 $\overline V^1=V^1\ominus[i_{01}(SH^1)]$
  ($({\dtwo}+I)\overline V^1=\overline V^1$) then all the mappings in
the above diagram become isomorphisms.
%We define {\em mono-dipole-free scalar} as an element of $\overline
%V^0$, {\em mono-dipole-free covector} belongs to $\overline V^1$ and
%any symmetric traceless tensor on $S^2$ is {\em mono-dipole-free}.
\subsubsection{Integral operators, generalized Green's functions}
Solution of the Helmholtz equation on a unit sphere $S^2$:
\[ \left[\dtwo +l(l+1)\right] \Psi_l(n) = \Phi(n) \, , \quad n\in S^2\]
is given (see e.g. \cite{RS-JMP47}) in terms of the generalized Green's function ${\bar G}_l$ as follows:
\be\label{rH} \Psi_l(n) = \int_{S^2} {\bar G}_l(n,n')\Phi(n')\dd^2n' \, .\ee
Here $n={\mathbf D}(\thet,\varphi)$ given by (\ref{dipolD}) and $\dd^2n=\mbox{\rm d}\sigma$.
The solution $\Psi_l(n)$ is automatically orthogonal to the space $SH^l$ (the kernel of Helmholtz operator $\dtwo +l(l+1)$) because Green's function is orthogonal to this space.
In our case we need to write the inverse of the operator ${\dtwo}({\dtwo}+2)$ as a double integral with the corresponding kernels ${\bar G}_l$ for $l=0$ and $l=1$.
More precisely, the solution $g$ of the equation ${\dtwo}({\dtwo}+2) g = f$ (with $P_{md} f=0$) is given in the following form:
\begin{eqnarray}\label{godwr} g(n) & = & \nonumber P_w
 \int_{S^2} {\bar G}_0(n,n") \cdot \\ \nonumber & &  \cdot \left[ \int_{S^2}{\bar G}_1(n",n')f(n')\dd^2n'
- \frac1{4\pi} \int_{S^2\times S^2}{\bar G}_1(m,n')f(n')\dd^2n'\dd^2 m\right] \dd^2n" \\  & = &
\int_{S^2} {\bar G}_0(n,n") \left[ \int_{S^2}{\bar G}_1(n",n')f(n')\dd^2n'
\right] \dd^2n"
 \, , \end{eqnarray}
where the projection operator $P_w$ provides orthogonality\footnote{The above integral operators do not mix the wave part with the mono-dipole part of a function. This means that $P_w f =f$ implies
$P_w (G_1*f)=G_1*f$ and $P_w (G_0*f)=G_0*f$.} of $g$ to the space $SH^0\oplus SH^1$.
The generalized Green's function written in a standard form:
\[ {\bar G}_l(n,n') = \sum_{i=0, i\neq l}^{\infty}\sum_{m=-i}^{i} \frac{Y_{im}(n)\overline{Y_{im}(n')}}{l(l+1)-i(i+1)} \, , \quad Y_{im}\in SH^i  \, , \]
can be simplified as follows (cf. \cite{RS-JMP47}):
\[ {\bar G}_l(n,n') = \frac1{4\pi} P_l(n \cdot n')\left[ \ln\frac{1-n\cdot n'}{2} +c_l \right]
 + \frac1{2\pi} \sum_{i=0}^{l-1} \frac{2i+1}{(l-i)(l+i+1)}P_i(n\cdot n') \, , \]
\[ c_l:=\frac1{2l+1}-2 \sum_{i=0}^{l-1} (-1)^{l+i}\frac{2i+1}{(l-i)(l+i+1)} \, , \]
$Y_{im}$ -- spherical harmonics (orthonormal basis in $SH^i$),
$n \cdot n'\in [-1,1]$ is a scalar product of unit vectors in ${\mathbb R}^3$
 and $\displaystyle P_l(x):=\frac1{2^l l!}(x^2-1)^{(l)}$ is the Legendre polynomial.

 \subsection{Second variation of area}

%%%%%%%%%%%%%%%%%%%%%%%%%appendix
The Gaussian coordinates (\ref{Gcoord}) and the definition of
the Riemann tensor gives
    \be
    R^u{}_{AuB}= \zk_{AB,u} + \zk_{A}{}^{C}\zk_{BC} \, .
    \ee
 This leads to
    \be
    k_{AB} = \zk_{AB} + \eps\tau (R^u{}_{AuB}  - \zk_{A}{}^{C}
    \zk_{BC}) + \eps\tau_{||AB} + O(\eps^2)\;.
    \ee
Taking the trace (and using (\ref{g-g0up})), we obtain:
    \be\label{varkR}
    k = \zk + \eps\tau (R^u{}_u + \zk^{AB}\zk_{AB})
    +\eps\tau^{||A}{_A} + O(\eps^2)\, .
    \ee
The formulae (\ref{varku}) and (\ref{varkR}) are equivalent
because of the Gauss-Codazzi equations:
\be 2\partial_u\zk = R(g_{kl}) - R(\zg_{AB}) + \zk^{AB}\zk_{AB}+\zk^2\, ,
\ee
\be 2R^u{}_u = R(g_{kl}) - R(\zg_{AB}) - \zk^{AB}\zk_{AB}+\zk^2\, ,
\ee
where $R(g_{kl})$ and $R(\zg_{AB})$ are scalar of curvatures of
the three-metric $g_{kl}$ and the two-metric $\zg_{AB}$, respectively.
Obviously (\ref{varku}), (\ref{varkR}) are the first
variations of the mean curvature $k$, which
 in the literature (see e.g. \cite{svararea})
  are known as the second variations of area.
They are usually presented in the following equivalent form:
\be\label{varkRR}
    k - \zk = \frac{\eps}2\left[  \left(R(g_{kl}) - R(\zg_{AB})
    + \zk^{AB}\zk_{AB}+\zk^2\right)\tau
    +2\tau^{||A}{_A}\right] + O(\eps^2)\, .
    \ee

\def\cprime{$'$}
\providecommand{\bysame}{\leavevmode\hbox
to3em{\hrulefill}\thinspace}
\providecommand{\MR}{\relax\ifhmode\unskip\space\fi MR }
\providecommand{\MRhref}[2]{%
  \href{http://www.ams.org/mathscinet-getitem?mr=#1}{#2}
} \providecommand{\href}[2]{#2}

\end{document}